%% file: main.tex
\documentclass{llncs}
\usepackage{amssymb}
\usepackage{graphicx}
\usepackage{subfigure}
\usepackage{enumerate}
\usepackage{epsfig}
\usepackage{xcolor}
\usepackage{amsmath}
\usepackage{hyperref}
\usepackage[leftcaption	]{sidecap}

\usepackage{xspace}
\usepackage{algorithm}

\usepackage{xcolor}
\usepackage{xkeyval}

\def\oh#1{O(#1)}
\def\set#1{\left\{ #1 \right\}}
\def\union{\bigcup}
\def\abs#1{\left| #1 \right|}
\def\ceil#1{\left\lceil #1 \right\rceil}

\newcommand{\alg}{{\ensuremath{\mathbb{A}}}\xspace}
\newcommand{\SEQ}[1]{\langle #1 \rangle}
\newcommand{\A}{\scriptsize{\ensuremath{\mathcal{ALG}}}\normalsize\xspace}

\newcommand{\ALG}{\scriptsize{\ensuremath{\mathcal{ALG}}}\xspace\normalsize}
\newcommand{\OPT}{\scriptsize{\ensuremath{\mathcal{OPT}}}\normalsize\xspace}
\newcommand{\opt}{\scriptsize{\ensuremath{\mathcal{OPT}}}\normalsize\xspace}

\newcommand{\ks}{k-{\sc Server }}
\newcommand{\kc}{k-{\sc Server }}

\newcommand{\pcv}{Path-Cover\xspace}
\newcommand{\npc}{\scriptsize{\ensuremath{\mathcal{GPC}}}\normalsize\xspace}

\newcommand{\sol}{\scriptsize{\ensuremath{\mathcal{PERM}}}\normalsize\xspace}

\newcommand{\pc}{PathCover\xspace}
\newcommand{\Tau}{\scriptsize{\ensuremath{\mathcal{T}}}\normalsize\xspace\xspace}
\newcommand{\Bita}{B\xspace}

\newcommand{\ca}{competitive analysis}
\newcommand{\Ca}{Competitive analysis}
\newtheorem{coro}{Corollary}

\title{On Advice Complexity of the $k$-server Problem under Sparse Metrics}
\author{Sushmita Gupta\inst{1} \and Shahin Kamali\inst{2} \and Alejandro L\'{o}pez-Ortiz\inst{2}}
\institute{University of Southern Denmark, Odense, Denmark
\and Cheriton School of Computer Science, University of Waterloo, Ontario, Canada}

\begin{document}
\maketitle
\input{body}

\bibliographystyle{splncs03}
\bibliography{refs}
\newpage

\end{document}

%% file: body.tex
\begin{abstract}
We consider the {\ks}problem under the advice model of computation when the underlying metric space is sparse. On one side, we show that an advice of size $\Omega(n)$ is required to obtain a $1$-competitive algorithm for sequences of size $n$, even for the 2-server problem on a path metric of size $N \geq 5$. Through another lower bound argument, we show that at least $\frac{n}{2}(\log \alpha- 1.22)$ bits of advice is required to obtain an optimal solution\footnote{We use $\log n$ to denote $\log_2(n)$.} for metric spaces of treewidth $\alpha$, where $4 \leq \alpha < 2k$. On the other side, we introduce $\Theta(1)$-competitive algorithms for a wide range of sparse graphs, which require advice of (almost) linear size. Namely, we show that for graphs of size $N$ and treewidth $\alpha$, there is an online algorithm which receives $\oh{n (\log \alpha + \log \log N)}$ bits of advice and optimally serves a sequence of length $n$. With a different argument, we show that if a graph admits a system of $\mu$ collective tree $(q,r)$-
spanners, then there is a $(q+r)$-competitive algorithm which receives $\oh{n (\log \mu + \log \log N)}$ bits of advice. Among other results, this gives a $3$-competitive algorithm for planar graphs, provided with $\oh{n \log \log N}$ bits of advice. 
\end{abstract}


\section{Introduction}\label{secIntro}
Online algorithms have been extensively studied in the last few decades. In the standard setting, the input to an online algorithm is a sequence of \textit{requests}, which should be \textit{answered} sequentially. To answer each request, the algorithm has to take an irreversible decision without looking at the incoming requests. For minimization problems, such a decision involves a \textit{cost} and the goal is to minimize the total cost. 

The standard method for analysis of online algorithms is the \textit{competitive analysis}, which compares an online algorithm with an optimal offline algorithm, \OPT. The competitive ratio of an online algorithm\footnote{In this paper we only consider deterministic algorithms.} \ALG is defined as the maximum ratio between the cost of \ALG for serving a sequence and the cost of \OPT for serving the same sequence, within an additive constant factor. 

Although the {\ca} is accepted as the main tool for analysis of online algorithms, its limitations have been known since its introduction; the prime critique being that it is not a good predictor of commonly observed behaviour.
{\Ca}' limitations go to the very heart of the technique itself. Inputs adversarially produced to draw out the worst performance of a particular algorithm are not commonplace in real life applications. Therefore, in essence {\ca} mostly measures the benefit of knowing the future, and not the true difficulty of instances. From the perspective of an online algorithm, the algorithm is overcharged for its complete lack of knowledge about the future.
Advice complexity quantifies this {\it gap in information} that gives \OPT an unassailable advantage over any online strategy. 

Under the advice model for online algorithms \cite{Emek2011,Bock11}, the input sequence $\sigma = \SEQ{r_1 \ldots r_n }$ is accompanied by $b$ bits of advice recorded on an advice tape $\Bita$. For answering the request $r_i$, the algorithm takes an irreversible decision which is a function of $r_1 , \ldots, r_i$ and the advice provided on $\Bita$.  The advice complexity of an online problem is the minimum number of bits which is required to optimally solve any instance of the problem.  In the context of the communication complexity, it is desirable to provide an advice of small size, while achieving high quality solutions. 

We are interested in the advice complexity of the $k$-{\sc Server} problem, as well as the relationship between the size of advice and the competitive ratio of online algorithms. 
To this end, we study the problem for a wide variety of sparse graphs. 

\subsection{Preliminaries}
An instance of the \kc problem includes a metric space $M$, $k$ mobile servers, and a request 
sequence $\sigma$. The metric space can be modelled as an undirected, weighted graph of size $N>k$ (we interchangeably use terms `metric space' and `graph'). Each request 
in the input sequence $\sigma$ denotes a vertex of $M$, and an online algorithm should move one of the servers to the requested vertex to 
\textit{serve} the request. The cost of the algorithm is defined as the total distance moved by all $k$ servers over $\sigma$. 

For any graph $G=(V,E)$, a tree decomposition of $G$ with width $\alpha$ is a pair $(\set{X_i \| i\in I \\ }, T)$ where $\set{X_i \| i\in I}$ is a family of subsets of $V$ (bags), and $T$ is a rooted tree whose nodes are the subsets $X_i$ such that
\begin{itemize}
	\item $\union_{i\in I} X_i = V$ and $\,\max\limits_{i\in I} \abs{X_i} = \alpha \!+\!1.$ 
	\item for all edges $(v,w)\in E$, there exists an $i\in I$ with $v\in X_i$ and $w\in X_i$. 
	\item for all $i,j,k \in I$: if $X_j$ is on the path from $X_i$ to $X_k$ in $T$, then $X_i \cap X_k \subseteq X_j$.
\end{itemize}

The treewidth of a graph $G$ is the minimum width among all tree decompositions of $G$. 
Informally speaking, the tree decomposition is a mapping a graph to a tree so that the vertices associated to each node (bag) of the tree are close to each other, and the treewidth measures how close the graph is to such tree. 

We say that a graph $G = (V,E)$ admits a system of $\mu$ collective tree $(q,r)$-spanners if there is a set $\Tau(G)$ of at most $\mu$ spanning trees of $G$ such that for any two vertices $x$, $y$ of $G$, there exists a spanning tree $T \in \Tau(G)$ such that $d_T(x, y) \leq q \times d_G(x, y) + r$.

For the ease of notation, we assume $k$ denotes the number of servers, $N$, the size of metric space (graph), $n$, the length of input sequence, and $\alpha$, the treewidth of the metric space. 

\subsection{Existing Results}

The advice model for the analysis of the online algorithm was first proposed in \cite{Emek2011}. Under that model, each request is accompanied by an advice of fixed length. A slight variation of the model was proposed in \cite{ISSAC09,Bock11}, which assumes that the online algorithm has access to an advice tape. At any time step, the algorithm may refer to the tape and read any number of advice bits. 
The advice-on-tape model has an advantage that enables algorithms to use sublinear advice (an advantage over the original model). This model has been used to analyze the advice complexity of many online problems, which includes paging~\cite{ISSAC09,MFCS10,SOFSEM11}, 
disjoint path allocation~\cite{ISSAC09}, job shop scheduling~\cite{ISSAC09,SOFSEM11}, $k$-server~\cite{Bock11}, knapsack~\cite{LATIN12}, bipartite graph coloring \cite{COCOON12}, online coloring of paths~\cite{LATA12}, 
set cover~\cite{CSR12,ECCC12}, maximum clique~\cite{ECCC12}, and graph exploration~\cite{SIROCCO12}. In this paper, we adopt this definition of the advice model.

For the \kc problem on general metrics, there is an algorithm which achieves a competitive ratio of $k^{\oh{1/b}}$ for $b \leq k$, when provided with $bn$ bits of advice \cite{Emek2011}. This ratio was later improved to $2  \lceil \lceil \log k \rceil / (b-2) \rceil $ in \cite{Bock11}, and then to $\lceil \lceil \log k \rceil / (b-2)  \rceil$ in \cite{WAOA11}. Comparing these results with the lower bound $k$ for the competitive ratio of any online algorithm \cite{Mana88}, one can see how an advice of linear size can dramatically improve the competitive ratio. 

Since the introduction of the \kc problem, there has been an interest in studying the problem under specific metric spaces. These metric spaces include trees \cite{Chrob91}, metric spaces with $k+2$ points \cite{Bart00}, Manhattan space \cite{Bein02}, the Euclidean space \cite{Bein02}, and the cross polytope space (a generalization of a uniform space) \cite{Bein07}. For trees, it is known that the competitive ratio of any online algorithm is at least $k$, while there are online algorithms which achieve this ratio \cite{Chrob91}. Under the advice model, the \kc problem has been studied when the metric space is the Euclidean plane, and an algorithm with constant competitive ratio is presented, which receives $n$ bits of advice for sequences of length $n$ \cite{Bock11}. In \cite{WAOA11}, tree metric spaces are considered and a $1$-competitive algorithm is introduced which receives $2 n + 2 \lceil \log (p+2) \rceil n$ bits of advice, where $p$ is the \textit{caterpillar dimension} of the tree. There are trees for which $p$ is as large as $\lceil \log N \rceil$. Thus, the $1$-competitive algorithm of \cite{WAOA11} needs roughly $2 \lceil n\log \log N \rceil$ bits of advice. Also in \cite{WAOA11}, it is proved that $\oh{n}$ bits of advice are sufficient to obtain an optimal algorithm for paths.

\subsection{Contribution}
Our first result implies that a sublinear advice does not suffice to provide close-to-optimal solution, even if we restrict the problem to 2-server problem on paths of size $N \geq 5$. Precisely, we show that $\Omega(n)$ bits of advice are required to obtain a $c$-competitive algorithm for any value of $c\leq 5/4 - \epsilon$ ($\epsilon$ is an arbitrary small constant). Since there is a 1-competitive algorithm which receives $\oh{n}$ bits of advice for paths \cite{WAOA11}, we conclude that $\Theta(n)$ bits of advice are necessary and sufficient to obtain a 1-competitive algorithm for these metrics.
Through another lower bound argument, we show that any online algorithm requires an advice of size at least $\frac{n}{2}(\log \alpha- 1.22)$ bits to be optimal on a metric of treewidth $\alpha$, where $4 \leq \alpha < 2k$.

On the positive side, we introduce an online algorithm which optimally serves any input sequence, when provided $\oh{n (\log \alpha + \log \log N)}$ bits of advice. For graphs with constant treewidth, the advice size is almost linear. Considering that an advice of linear size is required for 1-competitive algorithms (our first lower bound), the algorithm has an advice of nearly optimal size.
For graphs with treewidth $\alpha \in \Omega(\lg N)$
the advice size is $\oh{n \log \alpha}$, which is asymptotically tight when $4 \leq \alpha < 2k$, because at least $\frac{n}{2}(\log \alpha- 1.22)$ bits are required to be optimal in this case (our second lower bound). 

We provide another algorithm that achieves a competitive ratio of at most $q+r$, when the metric space admits a system of $\mu$ collective tree $(q,r)$-spanners. This 
algorithm receives $\oh{n (\log \mu + \log \log N)}$ bits of advice. This yields competitive algorithms for a large family of graphs, e.g., a 3-competitive algorithm for planar graphs, which reads $\oh{n (\log \log N)}$ bits of advice.


\section{Lower Bounds}

\subsection{2-server Problem on Path Metric Spaces}
In this section, we show that an advice of sublinear size does not suffice to achieve close-to-optimal solutions, even for the 2-server problem on a path metric space of size $N \geq 5$. 
Without loss of generality, we only consider online algorithms which are \textit{lazy} in the sense that they move only one server at the time of serving a request. 
It is not hard to see that any online algorithm can be converted to a lazy algorithm without an increase in its cost. Hence, a lower bound for the performance of lazy algorithms applies to all online algorithms. In the reminder of this section, the term {\it online algorithm} means a lazy algorithm.

Consider a path of size $N \geq 5$ which is horizontally aligned and the vertices are indexed from $1$ to $N$. Assume that the servers are initially positioned at vertices $2$ and $4$. We build a set of instances of the problem, so that each instance is formed by $m = n/7$ \textit{round}s of requests. Each round is defined by requests to vertices $(3, 1|5, 3, 2, 4, 2, 4)$, where the second request of a round can be either to vertex $1$ or vertex $5$. Each round ends with consecutive requests to vertices 2 and 4. So, it is reasonable to move servers to these vertices for serving the last requests of each round. This intuition is formalized in the following lemma.

\begin{lemma}
Consider an algorithm $A$ that serves an instance of the problem as defined above. There is another algorithm $A'$ with a cost which is not more than that of $A$, for which the servers are positioned at vertices $2$ and $4$ before starting to serve each round.
\end{lemma}

\begin{proof}
Consider the first round $R_t$ such that $A$ does not have a server positioned at vertex $2$ before serving the requests in $R_t$ (since the last request of the previous round has been to vertex $4$, there is necessarily  a vertex located at vertex 4). This implies that the last four requests of the previous round $R_{t-1}$ (to vertices $2,4,2,4$) are served by the same server $s_1$. So, $A$ pays a cost of at least 6 for serving these requests. Consider an algorithm $A'$, which moves the same servers as $A$ does for serving all requests before the last four requests of $R_{t-1}$. To serve the last four requests of $R_{t-1}$, $A'$ moves the servers to vertices $2$ and $4$. This requires a cost of at most $3$ (The worst case happens when the servers are positioned at 4 and 5 before the first request to 2; in this case the algorithm pays a cost of 2 to move the left server to position 2 and a cost of 1 to move the right server to position 4 on the next request). Hence, $A'$ pays a cost of at most three for the last four requests in $R_{t-1}$ and, when compared to $A$, saves a cost of at least 3 units in round $R_{t-1}$. At the beginning of round $R_{t}$, the servers of $A$ are positioned at $4$ and $x$ ($x\notin\{2,4\}$), and the servers of $A'$ are at $2$ and $4$. 
In future rounds, $A'$ moves the server positioned at 2 (resp. 4) in the same way that $A$ moves the server position at $x$ (resp. 4). The total cost would be the same for both algorithms, except that the cost for the first request which is served by the server positioned at $x$ in $A$ might be at most 3 units more when served by the server positioned at vertex 2 in $A'$. This is because the distance from vertex $2$ to any other vertex is at most $3$.  

To summarize, when compared to $A$, $A'$ saves a cost of at least $3$ units on the requests in the $R_{t-1}$ and pays an extra cost of at most $3$ for the rounds after $R_{t-1}$. Hence, the cost of $A'$ is not more than that of $A$. To prove the lemma, it suffices apply the above procedure on all rounds for which there is no server at position 2 at the beginning of the round. The result would be an algorithm which has servers located at positions 2 and 4 before the start of any round. \qed
\end{proof}

According to the above lemma, to provide a lower bound on the performance of online algorithms, we can consider only those algorithms which keep servers at positions $2$ and $4$ before each round.
For any input sequence, we say a round $R_t$ has type 0 if the round is formed by requests to vertices $(3, 1, 3, 2, 4, 2, 4)$ and has type $1$ otherwise, i.e., when it is formed by requests to vertices $(3, 5, 3, 2, 4, 2, 4)$.
The first request of a round is to vertex $3$. Assume the second request is to vertex $5$, i.e., the round has type 1. An algorithm can move the left vertex $s_l$ positioned at $2$ to serve the first request (to vertex 3) and the right server $s_r$ positioned at 4 to serve the second request (to vertex 5). For serving other requests of the round, the algorithm can move the servers to their initial positions, and pay a total cost of 4 for the round (see Figure \ref{fig:rightGuess}). Note that this is the minimum cost that an algorithm can pay for a round. This is because there are four requests to distinct vertices and the last two are request to the initial positions of the servers (i.e., vertices 2 and 4).

Now assume that the algorithm moves the right vertex $s_r$ to serve the first request (to vertex 3). The algorithm has to serve the second request (to vertex 5) also with $s_r$. The third request (to vertex 3) can be served by any of the servers. Regardless, the cost of the algorithm will not be less than 6 for the round (see Figure \ref{fig:wrongGuess}). 
With a symmetric argument, in case the second request is to vertex $1$ (i.e., the round has type 0), if an algorithm moves the right server to serve the first request it can pay a total cost of 4, and if it moves the left server for the first request, it pays a cost of at least 6 for the round. 

In other words, an algorithm should `guess' the type of a round at the time of serving the first request of the round. In case it makes a right guess, it can pay a total cost of 4 for that round, and if it makes a wrong guess, it pays a cost of at least 6. This relates the problem to the \textit{Binary String Guessing Problem}. 

\begin{figure}[!t]
\centering
 \subfigure[In case of a right guess for the type of a round, the algorithm can pay a cost of 4.]{\includegraphics[width=0.4\columnwidth]{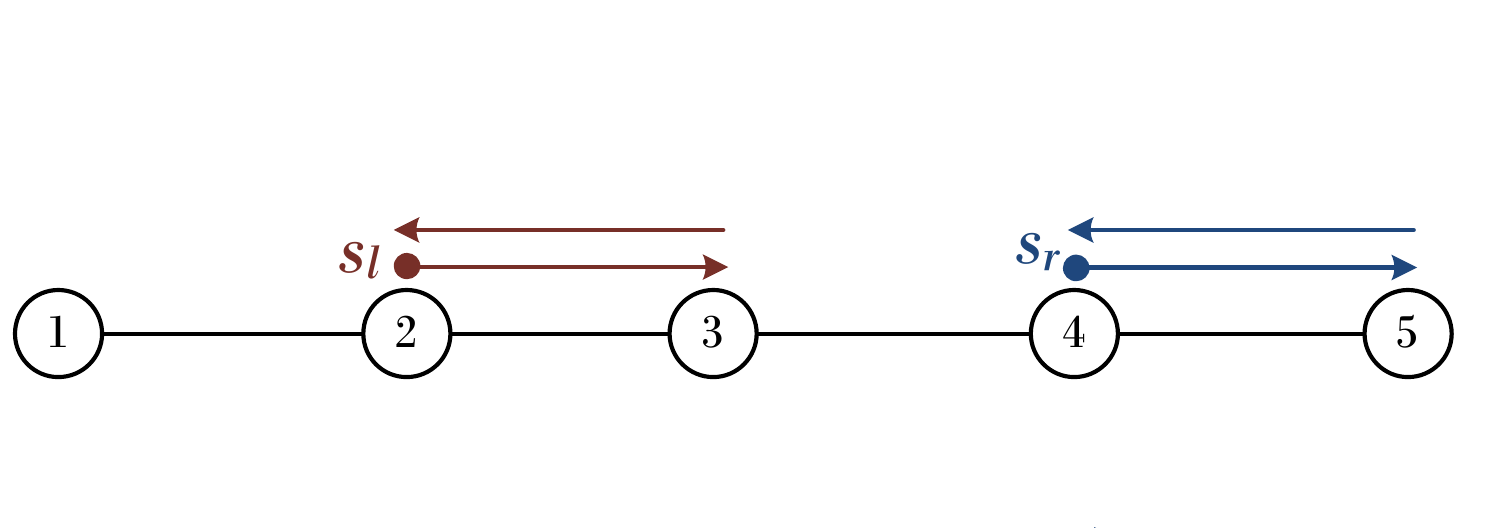}\label{fig:rightGuess}}
\hspace{.095\columnwidth}
 \subfigure[In case of a wrong guess for the type of a round, the algorithm pays a cost of at least 6.]{\includegraphics[width=0.4\columnwidth]{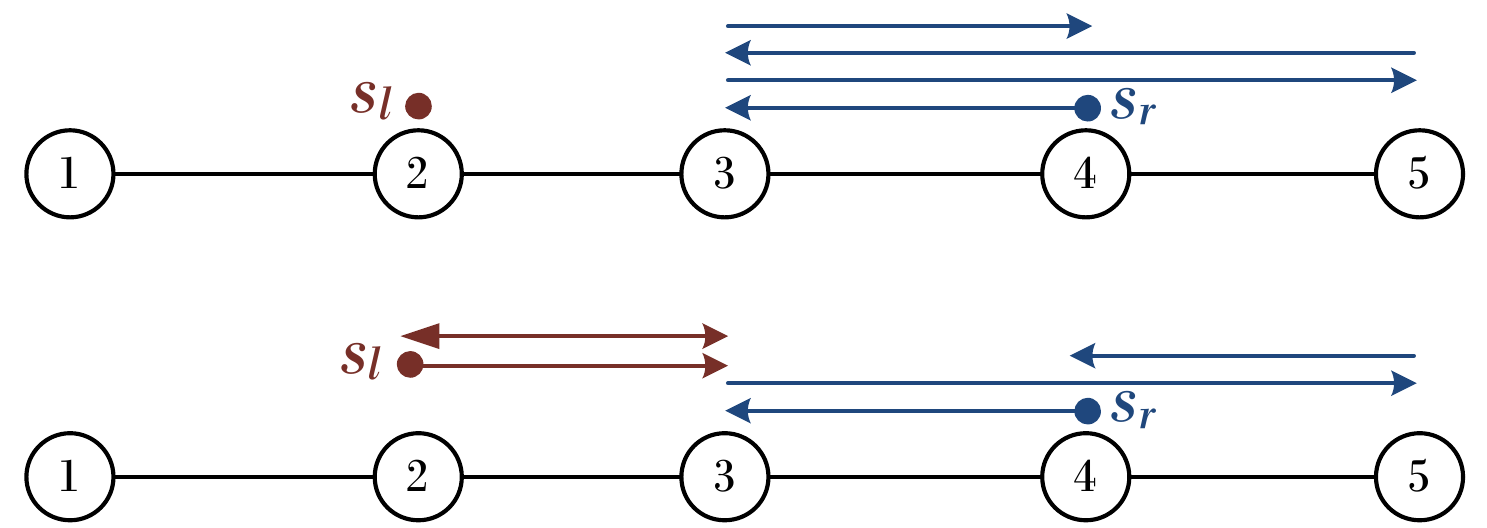}\label{fig:wrongGuess}}
\caption{The cost of an algorithm for a round of type 1, i.e., request to vertices $(3, 5, 3, 2, 4, 2, 4)$. The servers are initially located at 2 and 4.}
\label{packs}%
\end{figure}

\begin{definition}[\cite{Emek2011,ECCC12}]
The {\em Binary String Guessing Problem with known history ($2$-SGKH)} is the following online problem. The input is a bitstring of size $n’$, and the bits are revealed one by one. For each bit $b_t$, the online algorithm \alg must guess if it is a $0$ or a $1$. After the algorithm has made a guess, the value of $b_t$ is revealed to the algorithm.
\end{definition}

\begin{lemma}[\cite{ECCC12}] \label{servi}
On an input of length $m$, any deterministic algorithm for $2$-SGKH that is guaranteed to guess correctly on more than $\alpha m$ bits, for $1/2 \leq \alpha < 1$, 
needs to read at least $(1 + (1 - \alpha) \log(1 - \alpha) + \alpha \log \alpha) m$ bits of advice.
\end{lemma}

We reduce the $2$-SGKH problem to the $2$-server problem on paths. 
\begin{lemma}\label{redu1}
If there is a 2-server algorithm with cost at most $\gamma n$ ($\gamma \geq 4/7$) for an instance of length $n$ (as defined earlier), then 
there is a 2-SKGH algorithm which guesses at least $\frac{6-7\gamma}{2} m$ bits correctly for any input bit string of size $m = n/7$.
\end{lemma}

\begin{proof}
Let $B$ denote a bit string of length $m=n/7$, which is the input for the 2-SKGH problem. Consider the instance of the $2$-server problem in which the types of rounds are defined by $B$. Precisely, the $t$'th round has type 0 if the $t$'th bit of $B$ is 0, and has type 1 otherwise. We run the $2$-server algorithm on such an instance. At the time of serving the first request of the $t$'th round, the $2$-server algorithm guesses the type of round $t$ by moving the left or right server. In particular, it guesses the type of the round to be 0 if it moves the right server for the first request, and 1, otherwise. Define a $2$-SGKH algorithm which performs according to the 2-server algorithm, i.e., it guesses the $t$'th bit of $B$ as being 0 (resp. 1) if the 2-server algorithm guesses the $t$'th round as having type 0 (resp. 1). As mentioned earlier, the $2$-server algorithm pays a cost of 4 for the round for each right guess, and pays cost of at least 6 for each wrong guess. So, the cost of the algorithm is at least $4 \beta m + 6 (1-\beta) m = (6 - 2\beta)m$, in which $\beta m$ is the number of correct guesses ($\beta \leq 1$). Consequently, if an algorithm has cost at most equal to this value, it correctly guesses the types of at least $\beta m$ rounds, i.e., it correctly guesses at least $\beta m$ bits of a bit string of length $m$. Defining $\gamma$ as $(6 - 2\beta)/7$ completes the proof. \qed
\end{proof}

Lemmas \ref{servi} and \ref{redu1} give the following theorem.

\begin{theorem} \label{lowerMain}
On input of length $n$, any deterministic algorithm for the $2$-server problem which has a competitive ratio smaller than $\tau$ ($1< \tau < 5/4$), needs to read at least $(1+(2\tau-2) \log (2\tau-2) + (3-2\tau)\log(3-2\tau)) n/7$ bits of advice, even if the metric space is restricted to being a path of size $N \geq 5$. 
\end{theorem}

\begin{proof}
There is an offline 2-server algorithm which pays a cost of 4 for each round and consequently pays a total cost of $4m = 4n/7$. Hence, in order to have a competitive ratio of $\tau$, the cost of an algorithm should be at most $4\tau n /7$. According to Lemma \ref{redu1}, this requires the existence of a 2-SKGH algorithm which correctly guesses at least $(3-2\tau)m$ bits of a bit string of length $m$. By Lemma \ref{servi}, this needs reading at least $(1 + (1 - (3-2\tau)) \log(1 - (3-2\tau)) + (3-2\tau) \log (3-2\tau)) m = (1+(2\tau-2) \log (2\tau-2) + (3-2\tau)\log(3-2\tau)) n/7$ bits of advice. Note that $3-2\tau$ is in the range required by the lemma when $1< \tau < 5/4$. \qed
\end{proof}

For a competitive ratio of $\tau = 5/4$, the formula in Theorem \ref{lowerMain} takes the value 0 and thus does not provide a non-trivial bound. However, for doing strictly better than 5/4, a linear number of bits of advice is required. For example, to achieve a competitive ratio of $\tau = 6/5$, at least $.007262 n$ bits of advice are needed, and for the improved ratio of $\tau= 7/6$, at least $.020425 n$ bits of advice are needed. To achieve a $1$-competitive algorithm, $n/7$ bits of advice are required.

\subsection{Metrics with Small Treewidth}
We show that there are instances of the \ks problem in a metric space with treewidth $\alpha$, for which any online algorithm requires at least $\frac{n}{2}(\log \alpha- 1.22)$ bits of advice to perform optimally. Our construction is based on the one described in \cite{Bock11}, where a lower bound for a general metric space is provided. 

We introduce \textit{units graphs} and \textit{module graphs} as follows. A $\gamma$-unit graph is a bipartite graph $G = (U \cup  W, E)$ where $U=\{u_1, \ldots, u_\gamma\}$ contains $\gamma$ vertices, and $W$ contains $2^\gamma-1$ vertices each representing a proper subset of $U$. There is an edge between two vertices $u \in U$ and $w \in W$ iff $u \notin Set(w)$, where $Set(w)$ denotes the set associated with a vertex $w \in W$. 
Let $B_{i} \subseteq W$ denote the set of vertices of $W$ whose associate sets have size $i$. i.e., for $w \in B_i $ we have $|Set(w)| = i$. A \textit{valid request sequence} is defined as $\SEQ{x_{0}, x_{1}, \ldots , x_{\gamma-1}}$ so that for each $i$, $x_{i} \in B_{i}$ and $Set(x_{i}) \subseteq Set(x_{i+1})$. In other words, a valid sequence starts with a request to the vertex associated with the empty set, and with each step one element is added to get a larger set defining the next request. With this definition, one can associate every input sequence $I$ with a unique permutation $\pi$ of set $\{1, 2, \ldots, \gamma\}$.

A $\gamma$-module graph $G$ includes two $\gamma$-unit graphs $G_1 = (U_1 \cup  W_1, E_1)$ and $G_2 =(U_2 \cup  W_2, E_2)$. In such a graph, those vertices in $W_1$ which represent sets of size $i$ are connected to the $(i+1)$'th vertex of $U_2$; the vertices of $W_2$ and $U_1$ are connected in the same manner (see Figure \ref{fig:module}). Consider an instance of the \ks problem defined on a $k$-module graph, where initially all servers are located at the vertices of $U_1$. A valid sequence for the module graph is defined by repetition of \textit{rounds of requests}. Each round starts with a valid sequence for $G_1$ denoted by $\pi_1$, followed by $k$ requests to distinct vertices of $U_2$, a valid sequence for $G_2$, and $k$ requests to distinct vertices of $U_1$. It can be verified that there is a unique optimal solution for serving any valid sequence on $G$, and consequently a separate advice string is required for each sequence \cite{Bock11}. Since there are $(k!)^{(n/(2k))}$ valid sequences of length $n$, 
at least $(n/(2k)) \log (k!) \geq n(\log k-\log e)/2$ bits of advice are required to separate all valid sequences. 

The following lemma implies that the treewidth of the graphs used in the above construction is at most $2k$. 

\begin{figure}[t] 
\begin{minipage}{0.6 \linewidth}
\centering
\includegraphics[scale=0.4]{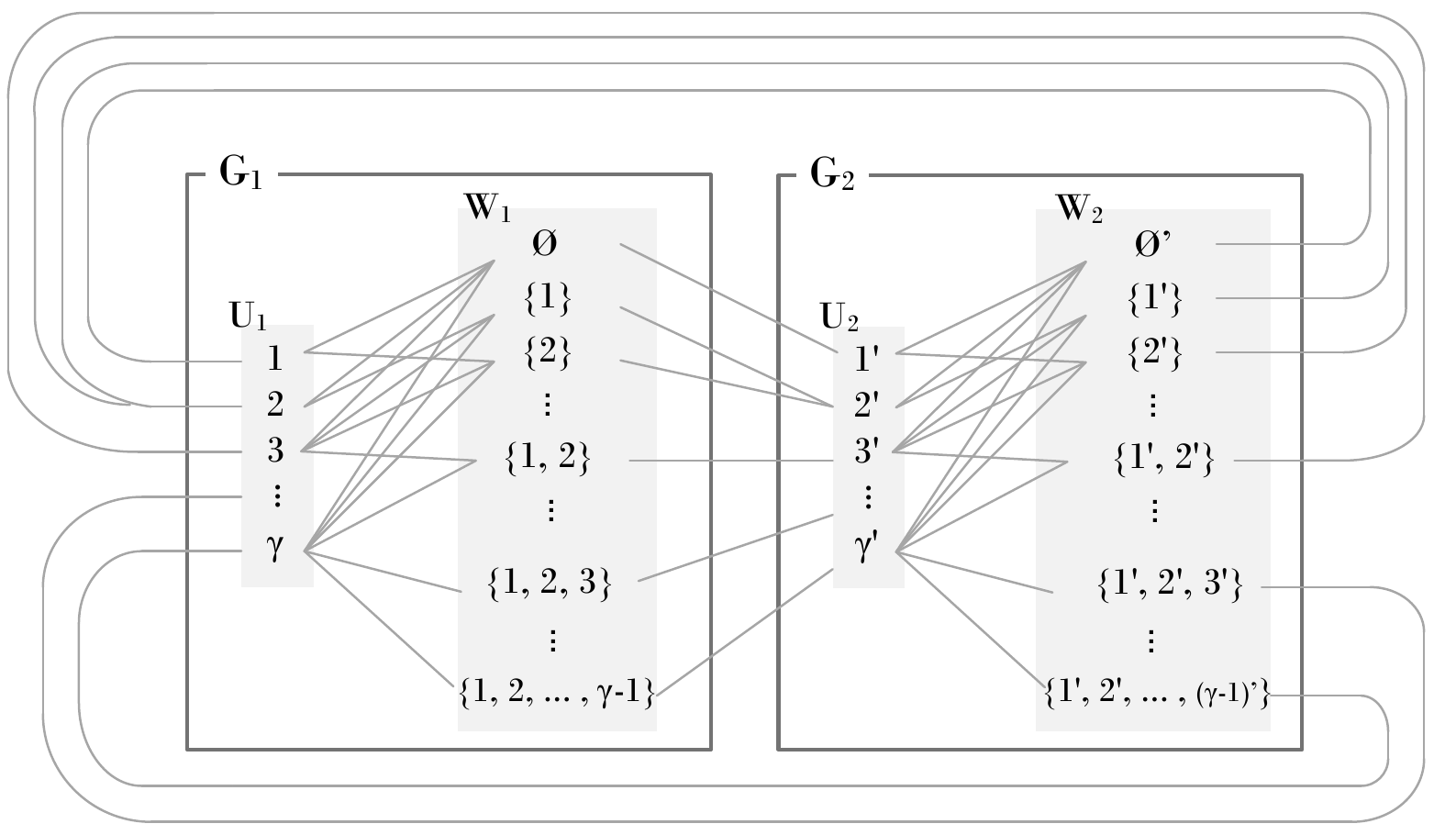}
\label{fig:module1}
\end{minipage}
\hspace{.2cm}
\begin{minipage}{0.4\linewidth}
\centering
\includegraphics[scale=0.4]{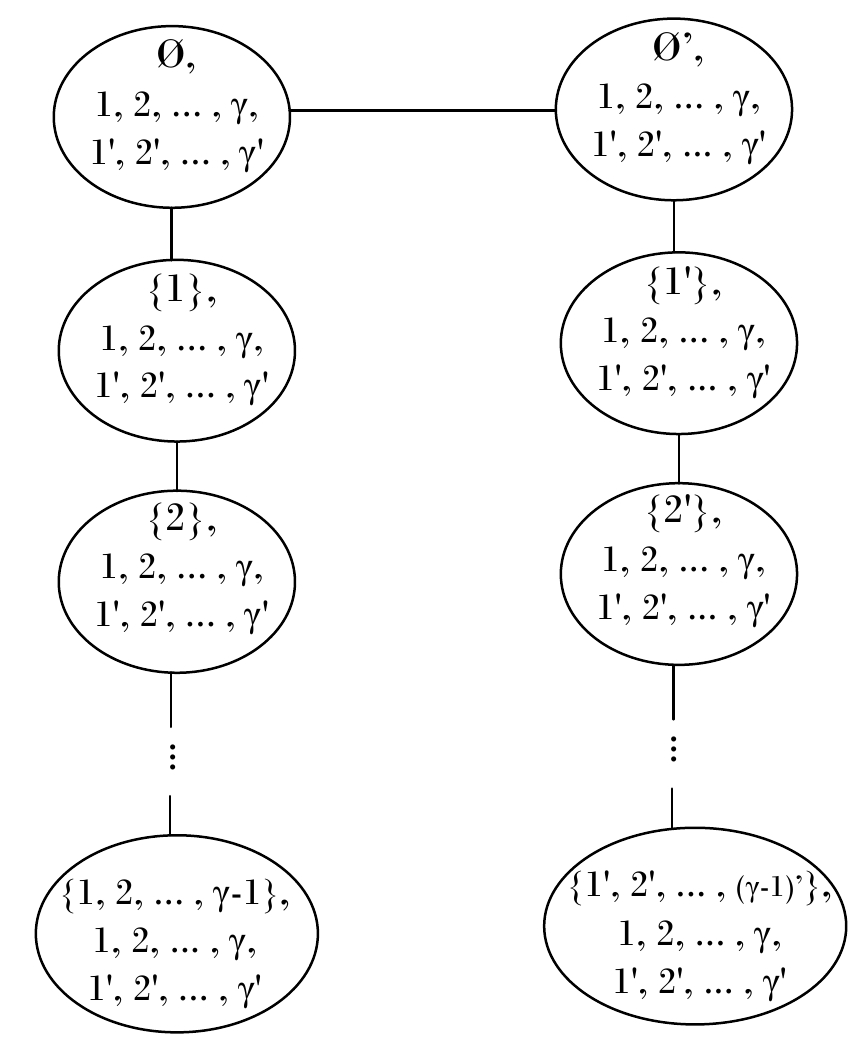}
\label{fig:module2}
\end{minipage}
\caption{A $\gamma$-module graph and a tree decomposition of it with treewidth $2 \gamma$. $G_1$ and $G_2$ are the unit graphs involved in construction of the module graph. \label{fig:module}}
\end{figure}

\begin{lemma} \label{mytd}
 Any $\gamma$-module graph has a tree decomposition of width $2\gamma$.
\end{lemma}

\begin{proof}
Let $G_1=(U_1 \cup  W_1, E_1)$ and $G_2=(U_2 \cup  W_2, E_2)$ be the unit graphs which define the $\gamma$-module graph. 
Define a tree decomposition as follows. Consider $2 \times 2^k$ bags so that each bag contains all vertices from $U_1$ and  $U_2$, and exactly one vertex from $W_1$ or $W_2$. Any tree which spans all these $2 \times 2^k$ bags is a valid tree decomposition (See Figure \ref{fig:module}). Moreover, there are exactly $2\gamma+1$ vertices in each bag which completes the proof.
\qed 
\end{proof}

For metrics with treewidth $\alpha \geq 2k$, the lower bound of $n (\log k-\log e)/2$ is tight, as $n \log k$ bits of advice are sufficient to serve each sequence optimally (by simply indicating the server that \opt would move to serve each request). In what follows, we consider metric spaces with treewidth $\alpha$ such that $4 \leq \alpha \leq 2k$. Assume that $\alpha$ is an even integer and we have $k = m \alpha /2$ for some positive integer $m$. 
Consider a metric space $G_{b}$ defined by a set of $\gamma$-modules where $\gamma = \alpha /2$. There are $k/ \gamma = m$ such modules in $G_b$. Let $M^1 , \ldots, M^{m}$ denote these modules, and let $G_1^i= (U_1^i\cup W_1^i, E_1), G_2^i= (U_2^i\cup W_2^i,E_2)$ denote the unit graphs involved in the construction of $M^i$ ($i \leq m$). For each module $M^i$, select exactly one vertex from $U_1^i$, and connect all of the selected vertices to a common \textit{source}. This makes $G_b$ a connected graph (see Figure \ref{fig:GB}).

\begin{figure}[t]
\begin{minipage}{0.5 \linewidth}
\centering
\includegraphics[width=0.8\columnwidth]{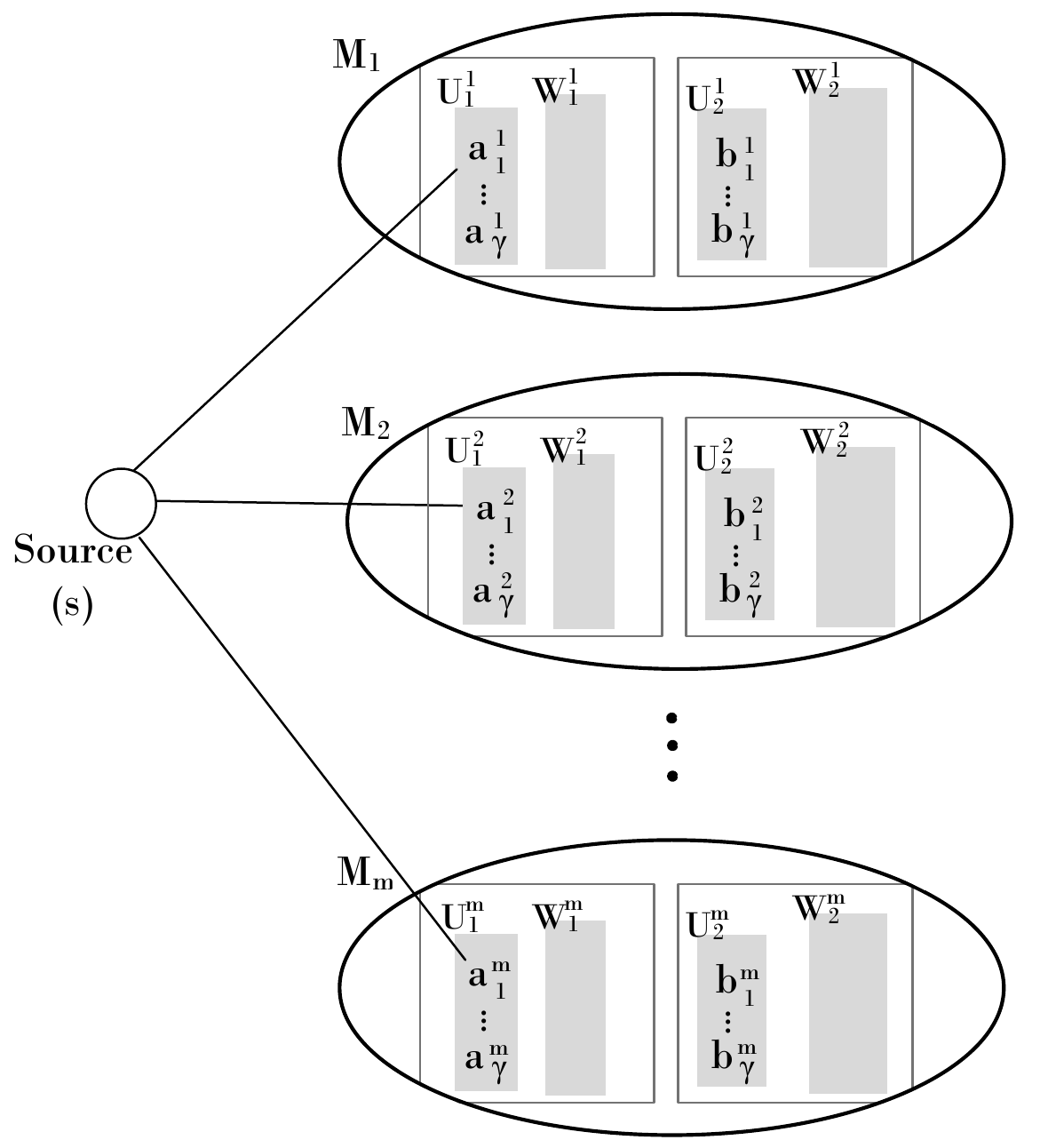}
\label{fig:figure1}
\end{minipage}
\begin{minipage}{0.5\linewidth}
\centering
\includegraphics[width=0.8\columnwidth]{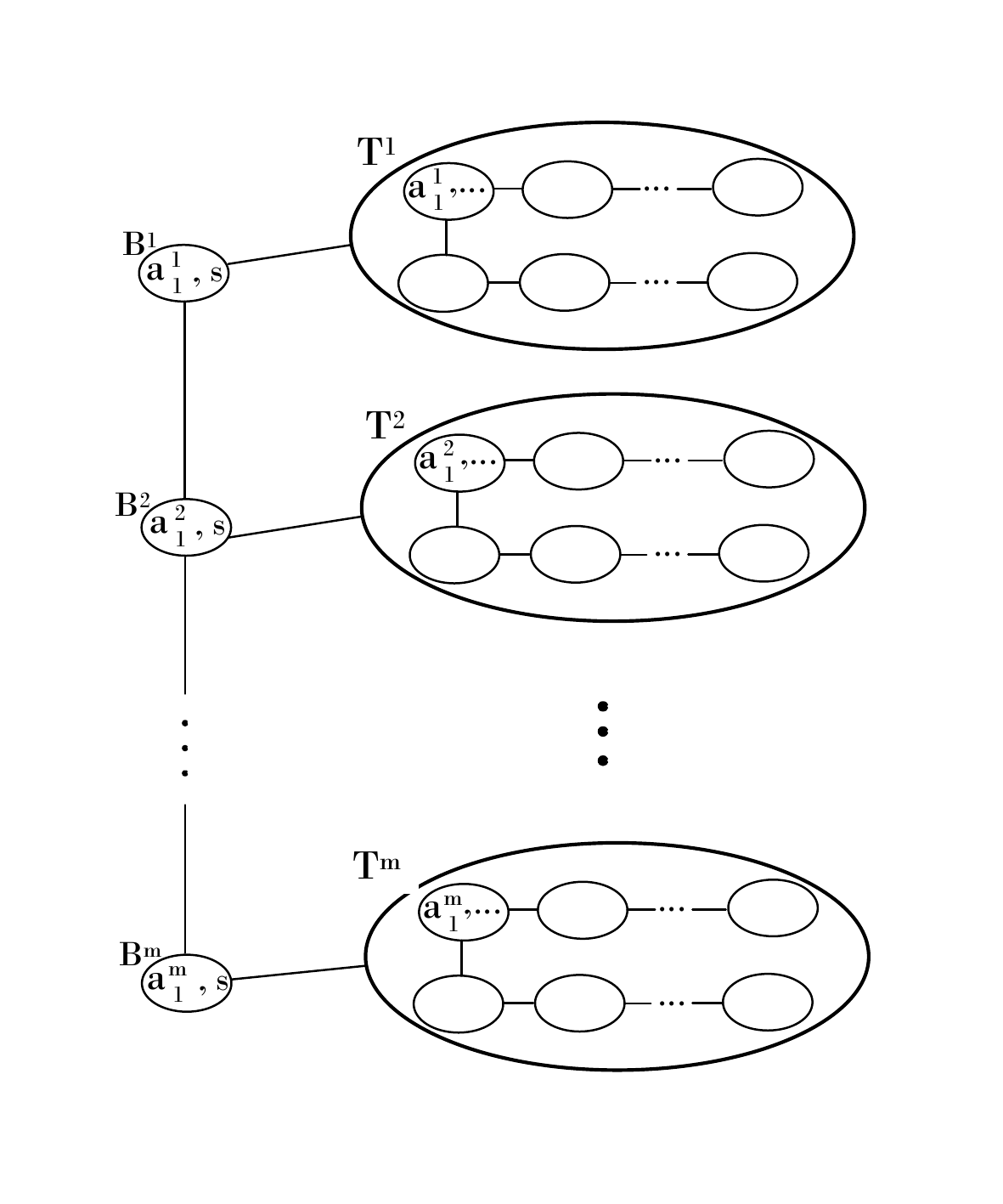}
\label{fig:figure2}
\end{minipage}
\caption{The metric space $G_b$ and a tree decomposition associated with it. The source $s$ is connected to the selected vertex $a^j_1$ of module $M^j$. \label{fig:GB}}
\end{figure}

\begin{lemma} \label{lemunlab}
The metric space $G_b$ has a tree decomposition of width $\alpha$.
\end{lemma}

\begin{proof}
By Lemma \ref{mytd}, each module has a tree decomposition of width $\alpha$. Let $T^i$ denote the tree associated with the decomposition of the $i$th module. For any tree $T^i$, consider a bag $B^i$ of 
size 2 which contains the source $s$ and the other endpoint of the edge between $s$ and $T^i$. Connect $B^i$ to an arbitrary bag of $T^i$. Add $m-1$ arbitrary edges between all $B^{i}$s to form a connected tree. Such a tree represents a valid tree decomposition of $G_b$ with width $\alpha$ \mbox{(see Figure \ref{fig:GB}).} \qed
\end{proof}

Since there are $m$ modules and in the $i$'th module $U_1^i$ contains $\gamma$ vertices, there are $m \times \gamma = k$ vertices in all of the $U_1^i$s. Assume that the $k$ servers are initially placed at separate nodes in the $U_1^i$s. A valid sequence for $G_b$ is defined by a sequence of rounds of requests in which each round has the following structure: \\

\footnotesize{\hspace{-.45cm}
$f(\pi_1^1, \ldots ,\pi_1^{m}) , (b_1^1, \ldots , b_1^{m}) , \ldots , (b_\gamma^1, \ldots , b_\gamma^{m}), f(\pi_2^1 , \ldots , \pi_2^{m}) ,
(a_1^1, \ldots , a_1^{m}), \ldots , (a_\gamma^1, \ldots , a_\gamma^{m})$}\\
\normalsize

Here, $f$ is a function that combines the requests from $m$ permutations. Let $(\pi^1, \ldots, \pi^m)$ denote $m$ permutations such that $\pi^i$ 
contains $\gamma$ requests $\SEQ{r^i_1, \ldots, r^i_\gamma}$ which defines a permutation in the module $M^i$. Thus, $f$ gives a sequence of length $ m \times \gamma$ starting with $m$ requests to $r^j_1$s, followed 
by $m$ requests to $r^j_2$s, and so on. For each $j$, ($1 \leq j \leq \gamma$) we have fixed orderings on the vertices such that $(a_{j}^{1}, \ldots, a_{j}^{m}) \in (U_{1}^{1}, \ldots, U_{1}^{m})$ and $(b_{j}^{1}, \ldots, b_{j}^{m}) \in (U_{2}^{1}, \ldots, U_{2}^{m})$. With this definition, when a valid sequence of $G_b$ is projected to the requests arising in a module $M$, the resulting subsequence is a valid sequence for $M$. 

\begin{lemma} \label{khast}
There is a unique optimal solution to serve a valid sequence on the metric space $G_b$. Also, each valid sequence requires a distinct advice string in order to be served optimally.
\end{lemma} 

\begin{proof}
 We present an algorithm \sol and show that its solution is the unique optimal solution for serving any valid sequence. To serve a request in $W_1^i$ (resp. $W_2^i$), \sol moves a server from $U_1^i$ (resp. $U_2^i$) to act according to the corresponding permutation $\pi_1^i$ (resp. $\pi_2^i$). To be precise, to serve a request $x_i$ (for $0 \leq i \leq \gamma-2$) it moves the server positioned at $Set(x_{i+1}) \setminus Set(x_{i})$, and thus leaving a unique choice for the last request $x_{\gamma-1}$. To serve a request $a_t^i$ in $U_1^i$ (resp. $U_2^i$), \sol moves the single server which is located at an adjacent node in $W_2^i$ (resp. $W_1^i$). Therefore, \sol incurs a cost of one for each request. 

There are $\gamma$ servers initially located in each $\gamma$ module, and \sol never moves a server from one module to another (no server passes the common source). To show that the solution of \sol is the unique optimal for serving any valid sequence, it is sufficient to show the following conditions: 

\begin{itemize}
\item No optimal algorithm moves a server from one module to another. 
\item Among all algorithms which do not move a server between modules, \sol provides the unique optimal solution.
\end{itemize}

Assume that there is an optimal algorithm \A which moves a server from one module to another. So at some point after serving the $t$'th request, there is a $\gamma$-module $M$ which has $\gamma + p$ servers stationed on it, for some $p \geq 1$. We show that the cost incurred by \A for serving the requests of $M$ in each round after the $t$'th request is lower bounded by $4\gamma - p$. Note that \sol incurs a total cost of $4\gamma$ in each round for the requests of any module. Assume that the cost incurred by  \A for serving a round in $M$ is strictly less than $4\gamma-p$. This implies that strictly more than $p$ requests incur no cost in that round. Since the same vertex is not requested twice in the same round, more than $p$ servers must not be moved in that round. So there is at least one server $s$ which is moved by \sol and not by \A. We show that moving $s$ in the same way as \sol does decreases the cost for \A. Assume that \A keeps $s$ at some vertex in $U_1$ of $M$. Thus, \sol moves $s$ from $U_1$ to serve a request in $W_1$, and then moves it to serve some request in $U_2$, followed by a move to serve a request in $W_2$, and finally a move to serve a request in $U_1$. Each of these moves cost one for \sol. However, each of the involved requests imposes a cost of 2 for \A since it has to use some server to serve at least two requests arising in $W_1$ in that round, thus, requiring a move via some vertex in $U_1$ or $U_2$. The same holds for the requests in $U_1$, $U_2$ and $W_2$. We can make similar arguments when \A keeps $s$ at some vertex in $W_1$ or $U_2$ of $M$, and conclude that \A saves a cost of $x$ by not moving $s$ in a round while incurring a cost of $2x$ in the remainder of the round. Hence, \A must incur a cost of at least $4\gamma-p$ for requests arising in $M$.

Let $M'$ be a $\gamma$ module in which \A has stationed $\gamma - q, ~(q \geq 1)$ servers after serving the $t$'th request. We show that the cost incurred by \A to serve the requests from $M'$ in each round starting after the $t$'th request is lower bounded by $4 \gamma + 8q$. Similar to the previous argument, since at least one server is missing, for some request(s) arising in $W_1$, \A has to use server(s) already located in $W_1$. So instead of incurring a cost of 1 as \sol, \A incurs a cost of 2 for each of those requests (the same holds for requests in $U_1$,$U_2$, and $W_2$).

To conclude, if \A moves $x$ servers between modules, compared to \sol, it saves at most  $x$ units of cost on the requests arising in the modules which receive extra servers, while it has to pay at least an extra $8x$ units for the requests in modules which lose their servers. This is in addition to the cost involved in moving servers from one module to another. We conclude that an optimal algorithm never moves servers between modules -condition 1.  

Inside each module, \sol acts the same as the unique optimal algorithm presented in \cite{Bock11}. Recall that the requests projected to each module form a valid sequence for that module, and can be treated independently (since servers do not move between modules in an optimal scheme). Hence, both conditions 1 and 2 are met, and \sol is the unique optimal for 
serving any valid request in $G_b$. 

Next, we show that each valid sequence requires a distinct advice string. Assume that two valid sequences $\sigma$ and $\sigma'$ differ for the first time at the $t$'th request. Note that two valid sequences of $G_b$ can only differ on the requests which define the permutations. Hence, $t$ should be a request belonging to $\pi^i_1$ or $\pi^i_2$ of some module $M^i$, i.e, one of the permutations representing a valid subsequence for the unit graphs defining $M^i$. Let $t_{0}<t$ denote the index of the previous request to an item in the same unit graph (that is, the previous request in the same permutation). While serving the request indexed $t_{0}$ in the two sequences,  an optimal algorithm will move different servers in anticipation of the $t$'th request. Hence an online algorithm should receive different advice strings to perform optimally for both sequences. \qed 
\end{proof}

To find a lower bound for the length of the advice string, we count the number of distinct valid sequences for the metric space $G_b$.
In each round there are $(\gamma !)^2$ valid sequences for each $\gamma$-module. Since there are $m$ such modules, there are $(\gamma !) ^ {2m}$ possibilities for each round. A valid sequence of length $n$ involves $n/(4\gamma m)$ rounds; hence there are $(\gamma !)^ {n/(2\gamma)}$ valid sequences of length $n$. Each of these sequences need a distinct advice string. Hence, at least $\log ((\gamma !)^{n/(2 \gamma)}) \geq (n/2) \log (\gamma /e) =   (n/2) \log (\alpha / (2e))$ bits of advice are required to serve a valid sequence optimally. This proves the following theorem.

\begin{theorem}\label{mainLowTreeWidth}
Consider the \ks problem on a metric space of treewidth $\alpha$, such that $4 \leq \alpha < 2k$. At least $\frac{n}{2}(\log \alpha- 1.22)$ bits of advice are required to optimally serve an input sequence of length $n$.
\end{theorem}
\section{Upper Bounds}

\subsection{Graphs with Small Treewidth}
We introduce an algorithm called \textit{Graph-Path-Cover}, denoted by \npc, to show that $\oh{n (\log \alpha + \log \log N)}$ bits of advice are sufficient to optimally serve a sequence of length $n$ on any metric space of treewidth $\alpha$. We start with the following essential lemma. 

\begin{lemma} \label{lemmaMMT} Let $T$ be a tree decomposition of a graph $G$. Also, let $x$ and $y$ be two nodes of $G$ and $P = (x=p_0,p_1, \dots p_{l-1}, y=p_l)$ be the shortest path between $x$ and $y$. Let $X$ and $Y$ be two bags in $T$ which respectively contain $x$ and $y$. Any bag on the unique path between $X$ and $Y$ in $T$ contains at least one node $p_i$ $(0\leq i\leq l$) from $P$.
\end{lemma} 

\begin{proof}
By the definition of the tree decomposition, each vertex $v$ of $G$ is listed in the bags of a contiguous subtree $T_v$ of $T$. Consider two vertices $p_i$ and $p_{i+1}$ in $P$. Since $p_i$ and $p_{i+1}$ are neighbors, there is a bag in $T$ which contains both of them. So the union of the subtrees $T_{p_i}$ and $T_{p_{i+1}}$ forms a (connected) subtree of $T$. Similarly, the union of all the subtrees of the nodes $p_0, \ldots, p_l$ form a (connected) subtree in $T$. Such a subtree contain $X$ and $Y$ and hence, any bag on the path between them. So any bag between $X$ and $Y$ contain at least one vertex $p_i$ of $P$.\qed 
\end{proof}

Similar to the \pcv algorithm introduced for trees in \cite{WAOA11}, \npc moves its servers on the same trajectories as \opt moves its. Suppose that \OPT uses a server $s_i$ to serve the requests $\left[ r_{a_{i,1}} , \ldots, r_{a_{i,n_i}}\right]$ ($i\leq k, n_i \leq n$). So, $s_i$ is moved on the unique path from its initial position to $r_{a_{i,1}}$, and then from $r_{a_{i,1}}$ to $r_{a_{i,2}}$, and so on. Algorithm \pcv tends to move $s_i$ on the same path as \opt.

For any node $v$ in $G$, \npc treats one of the bags  which contains $v$ as the \textit{representative bag} of $v$. Moreover, it assumes an ordering of the the nodes in each bag. Each node in $G$ is \textit{addressed} via 
its representative bag, and its index among the nodes of that bag. A server $s_i$, located at a vertex $v$ of $G$, is addressed via a bag which contains $v$ (not necessarily the representative bag of $v$) and the index of $v$ in that bag. Note that while there may be a unique way to address a node, there might be several different ways to address a server.

Assume that for serving a request $y$, \opt moves a server $s_i$ from a node $x$ to $y$ in $G$. Let $X$ and $Y$ be respectively the representative bags of $x$ and $y$,
and $Z$ be the least common ancestor of $X$ and $Y$ in $T$.  By Lemma \ref{lemmaMMT}, the shortest path between $x$ and $y$ passes at least one node $z$ in $Z$,  and 
that node can be indicated by $\lceil \log h \rceil + \lceil \log \alpha \rceil$ bits of advice ( $h$ denotes the height of the tree associated with the tree 
decomposition), with $\lceil \log h \rceil$ bits indicating $Z$ and $\lceil \log \alpha \rceil$ bits indicating the index of the said node $z$ in $Z$. After serving $x$, \npc moves $s_i$ to $z$, provided that the address of $z$ is given as part of the advice for $x$. For serving $y$, \npc moves $s_i$ to $y$, provided that the address of $s_i$ (address of $z$) is given as part of the advice for $y$. In what follows, we elaborate this formally.

Before starting to serve an input sequence, \npc moves each server $s_i$ from its initial position $x_0$ to a node $z_0$ on the shortest path between $x_0$ and the first node $r_{a_{i,1}}$ served by $s_i$ in \opt's scheme. \npc selects $z_0$ in a way that it will be among the vertices in the least common ancestor of the representative bags of $x_0$ and $r_{a_{i,1}}$ in the tree decomposition (by Lemma \ref{lemmaMMT} such a $z_0$ exists). To move all servers as described, \npc reads $(\lceil \log h \rceil + \lceil \log \alpha \rceil) \times k$ bits of advice. After these \textit{initial} moves, \npc moves servers on the same trajectories of \opt as argued earlier. 
Assume that $x$, $y$ and $w$ denote three requests which are consecutively served by $s_i$ in \opt's scheme. The advice for serving $x$ contains $\lceil \log h \rceil + \lceil \log \alpha \rceil$ bits which represents a node $z_1$,  which lies on the shortest path between $x$ and $y$ and is situated inside the least common ancestor of the respective bags in $T$. \npc moves $s_i$ to $z_1$ after serving $x$. The first part of advice for $y$ contains $\lceil \log h \rceil + \lceil \log \alpha \rceil$ bits indicating the node $z_1$ from which $s_i$ is moved to serve $y$. The second part of advice for $y$ indicates a node $z_2$ on the shortest path between $y$ and $w$ in the least common ancestor of their bags in $T$. This way, $2 (\lceil \log h \rceil + \lceil \log \alpha \rceil)$ bits of advice per request are sufficient to move servers on the same trajectories as \opt. 

The above argument implies that an advice of size  $2 (\lceil \log h \rceil + \lceil \log \alpha \rceil) \times n + (\lceil \log h \rceil + \lceil \log \alpha \rceil) \times k$ is sufficient to achieve an optimal algorithm. The value of $h$ (the height of the tree decomposition) can be as large as $N$, 
however we can apply the following lemma to obtain height-restricted tree decompositions.

\begin{lemma} \cite{Bodlaender93atourist,ICALP11} \label{lem:depthtc}
Given a tree decomposition with treewidth $\alpha$ for a graph $G$ with $N$ vertices, one can obtain a tree decomposition of $G$ with height $\oh {\log N}$ and width at most $3\alpha+2$.
\end{lemma}

If we apply \npc on a height-restricted tree decomposition, we get the following theorem.

\begin{theorem}\label{UpperBoundTreeWidth}
 For any metric space of size $N$ and treewidth $\alpha$, there is an online algorithm which optimally serves any input sequence of size $n$, provided with $\oh{n (\log \alpha + \log \log N)}$ bits of advice.
\end{theorem}

\subsection{Graphs with Small Number of Collective Tree Spanners}

In this section we introduce an algorithm which receives an advice of almost linear size and achieves constant competitive ratio for a large family of graphs. 


\begin{theorem}\label{thspan}
If a metric space of size $N$ admits a system of $\mu$ collective tree $(q,r)$-spanners, then there is a deterministic online algorithm which on receiving 
$ \oh{n \left(\log \mu + \log\log N\right)}$ bits of advice, achieves a competitive ratio of at most $q+r$ on any sequence of length $n$.
\end{theorem}

\begin{proof} 
When there is only one tree $T$ in the collection (i.e., $\mu = 1$), we can apply the \pc algorithm of \cite{WAOA11} on $T$ to obtain the desired result. To be precise, 
for the optimal algorithm \opt, we denote the path taken by it to serve a sequence of requests with the server $s_i$ to be $P_G= \left[  x_{a_{i,1}},x_{a_{i,2}}, \ldots, 
x_{a_{i,n_i}} \right]$.


\pc algorithm moves $s_i$ on the path $P_T = \left[ x_{a_{i,1}}, \ldots,x_{a_{i,2}}, \ldots, x_{a_{i,n_i}} \right]$ in $T$. Since 
$T$ is a spanner of $G$, the total length of $P_T$ does not exceed that of $P_G$ by more than a factor of $q+r$ for each edge in $P_G$, and consequently the cost of 
the algorithm is at most $q+r$ times that of {\opt}'s. Thus, the algorithm is $(q+r)$-competitive. 

After serving a request $x$ with server $s_i$, \pc can move $s_i$ to the least common ancestor of $x$ and $y$, where $y$ is the next request at which \opt uses $s_{i}$. This 
requires $\lceil \log h \rceil$ bits of advice per request ($h$ being the height of the tree). Instead, the algorithm can use the caterpillar decomposition of $T$ 
and move servers on the same set of paths while using only $\oh{\log\log N}$ bits of advice. The main idea is the same, whether we use a rooted tree or the caterpillar decomposition. Here for the ease of explanation, we will only argue for the rooted tree, but the 
statement of the theorem holds when the caterpillar decomposition is used. 

We introduce an algorithm that mimics  \opt's moves for each server,  by picking suitable trees from the collection to move the server through.The advice provided with each request indicates which tree from the collection would best approximate the 
edges traversed by the server in \OPT's scheme to reach the next node at which it is used. To this end, we look at the tree spanners as rooted trees. If \opt moves a server $s_i$ on the path $P_G = \left[  x_{a_{i,1}},x_{a_{i,2}}, \ldots, x{_{i,n_i}} \right]$, then for each edge $(x_{a_{i,j}},x_{a_{i,j+1}})$ on this path, our algorithm moves $s_i$ on the shortest path of (one of) the tree spanners which best 
approximates the distance between the vertices $x_{a_{i,j}}$ and $x_{a_{i,j+1}}$.  As explained below the selection of suitable spanners at every step can be ensured by providing $2\ceil{\log \mu}$ bits of advice with each request.

Let us denote the initial location of the $k$ servers by $z_{1}, \ldots, z_{k}$, and let $z'_{1}, \ldots, z'_{k}$ respectively denote the first requested nodes 
served by them. Before starting to serve the sequence, for any server $s_i$, the algorithm reads $\lceil \log \mu \rceil+ \lceil \log h \rceil$ bits of advice to detect the 
tree $T_p (1 \leq p \leq \mu)$ that preserves the distance between $z_i$ and $z'_i$ in $G$, and moves $s_i$ to the least common ancestor of $z_i$ and $z'_i$. 
Moreover, the algorithm labels $s_i$ with index $p$. These labels are used to move the correct servers on the trees in order to cover the same paths as \opt. Let $w$ and 
$y$ be two vertices which are served respectively before and after $x$ with the same server in \opt's scheme. To serve the request to $x$ the algorithm works as 
follows:

\begin{itemize}
	\item Find the spanner $T_{p}$ which best approximates the length of the shortest path between $w$ and $x$ in $G$. This can be done if provided with 
	$\lceil \log \mu \rceil$ bits of advice with $x$.
	\item Read $\lceil \log h \rceil$ bits of advice to locate a server $s$ labeled as $p$ on the path between node $x$ and the root of $T_{p}$. Move $s$ to serve 
	$x$. In case of caterpillar decomposition, the algorithm reads roughly $\log \log N$ bits.
	\item After serving $x$, find the spanner $T_{q}$ which best approximates the length of the shortest path between $x$ and $y$ in $G$. This can be done if provided with $\lceil \log \mu \rceil$ bits of advice with $x$.
	\item Find the least common ancestor of $x$ and $y$ in $T_{q}$. This can be done by adding $\lceil \log h \rceil$ bits of advice for $x$, where $h$ is the 
	height of $T_{q}$. In case of caterpillar decomposition, this would require roughly $\log \log N$ bits.
	\item Move $s$ to the least common ancestor of $x$ and $y$ and label it as $q$.
\end{itemize}

Thus, since \opt moves the server $s_i$ on the path $P_G = [x_{a_{i,1}},x_{a_{i,2}}, \ldots, x_{a{_{i,n_i}}}]$, our algorithm moves $s_i$ from $x_{a_{i,j}}$ to 
$x_{a_{i,j+1}}$ for each $j$ ($1 \leq j \leq n_i-1$), on the path in the tree which approximates the distance between these two vertices within a multiplicative factor 
of $q+r$. The labels on the servers ensure that the algorithm moves the `correct' servers on the trees. i.e, the ones which were intended to be used. Consequently, the
cost of an algorithm for each server is increased by a multiplicative factor, at most $q+r$. Therefore, the total cost of the algorithm is at most $(q+r) \times \opt$. 
The size of advice for each request is $2\lceil \log \mu \rceil + \oh{\log \log N }$, assuming that the caterpillar decomposition is used. Adding to that an additional  
$k (\log \mu  +  \oh{\log \log N})$ bits for the initial movement of servers completes the proof. \qed
\end{proof}

In recent years, there has been wide interest in providing collective tree spanners for various families of graphs, mostly in the context of message routing in networks. The algorithms which create these spanners run in polynomial time and in some cases linear time. For example, it any planar graph of size $N$ has a system of $\log N$ collective (3,0)-spanners \cite{pezdisp}; every AT-free graph (including interval, permutation, trapezoid, and co-comparability graphs) admits a system of two (1,2)-spanners \cite{WG04}; every chordal graph admits a system of at most $\log N$ collective (1,2)-spanners \cite{SWAT04}; and every Unit Disk Graphs admits a system of $2 \log _{1.5} n +2$ collective (3,12)-spanners \cite{DUGSPANNER}. 

\begin{coro}
For metric spaces of size $N$ and sequences of length $n$, $\oh{n \log \log N}$ bits of advice are sufficient to obtain I) a 3-competitive algorithm for planar graphs 
II) a 3-competitive algorithm for AT-free graph (including interval, permutation, trapezoid, and co-comparability graphs)
III) a 3-competitive algorithm for chordal graphs 
IV) a 15-competitive algorithm for Unit Disk Graphs.
\end{coro}

\section*{Concluding Remarks}
For path metric spaces, we showed any 1-competitive algorithm requires an advice of size $\Omega(n)$. This bound is tight as there is an optimal algorithm \cite{WAOA11} which receives $\oh{n}$ bits of advice.
The same lower bound applies for trees, however, the best algorithm for tree receives an advice of $\oh{n \lg \lg N}$. We conjecture that the lower bound argument can be improved for trees to match it with upper bound, and leave this as future work.